\documentclass{article}
\usepackage[utf8]{inputenc}
\usepackage{amsmath}
\usepackage{amsthm}
\usepackage{amsfonts}
\usepackage{mathtools}
\usepackage{amsmath,amssymb}
\usepackage[inline]{showlabels}
\usepackage{graphicx}
\usepackage[vcentermath]{youngtab}
\usepackage{mathtools}
\usepackage{latexsym,amssymb,amsmath,amsfonts, amsthm, xcolor}

\newtheorem{theorem}{Theorem}
\newtheorem{lemma}[theorem]{Lemma}

\newtheorem{remark}[theorem]{Remark}

\title{Entanglement of Quantum States which are Zero on the Symmetric Sector}

\author{D. D'Alessandro\thanks{Department of Mathematics, Iowa State University, Ames, IA 50011, daless@iastate.edu}}

\begin{document}

\maketitle

\begin{abstract}

We consider a quantum system of {\it $n$ qudits} and the Clebsch-Gordan decomposition of the associated Hilbert space $(\mathbb{C}^d)^{\otimes n}$. In this decomposition,  one of the subspaces is the so-called {\it  symmetric subspace} or {\it symmetric sector}, that is, the subspace of all states that  are invariant under the action of the symmetric group. We prove that  any separable state must have a nonzero component on the symmetric sector, or, equivalently, any state which has zero component on the symmetric sector must be entangled. For the cases of  $n=2,3$ particles, and in arbitrary dimension $d$, this 
 result can be refined  by providing sharp lower bounds on the size of the component of separable states on the symmetric sector. This leads us to identify a class of {\it entanglement witnesses} for these systems.  We provide an example showing that in the multipartite case, this  class of witnesses detects PPT entangled states. 

\end{abstract}

\vspace{0.5cm}


\vspace{0.5cm}

\section{Introduction}

The study of the geometry of quantum states is at the heart of quantum mechanics and quantum information theory. In this context, when considering multipartite quantum systems, the states are naturally classified in {\it separable} states, that is, {\it pure} states that can be written as products $|\psi\rangle=|\psi_1 \rangle  \otimes |\psi_2\rangle  \otimes \cdots \otimes |\psi_n\rangle $, where $|\psi_j\rangle$ is a state  of the $j$-th subsystem, and {\it entangled} states, that is, states that cannot be written as products.\footnote{The definition is a bit more elaborate for mixed states (see, e.g., \cite{NC}) and we shall consider it later in the paper.} Detecting 
 entangled states  and classifying them according to the type and amount of entanglement (using  a specific entanglement measure)  is one of the most studied and fundamental topics  in quantum information \cite{Hororev}. 

When considering a multipartite system of $n$, $d$-dimensional, subsystems ({\it qudits}), the Hilbert space ${\cal H}:=(\mathbb{C}^d)^{\otimes n}$,  splits according to the {\it Clebsch-Gordan (CG) decomposition} into irreducible representations of $SU(d)$. This decomposition, in the case of $d=2$ ({\it qubits}) is well studied in the theory of angular momentum addition (see, e.g., \cite{Sakurai}),  and it has several important applications in the general case  $d \geq 2$ (see, e.g., \cite{ArneAlex} and references therein). Irreducible representations of $SU(d)$ appear with various multiplicities in ${\cal H}$. One representation which always appears with multiplicity one is spanned by states which are invariant under the permutation group $S_n$. This is referred to as the {\it symmetric subspace} or {\it symmetric sector}. 

Symmetric states  have many applications in quantum information theory \cite{Harrow} and therefore they have attracted great attention in recent years both at the theoretical and at  the experimental level (see, e.g., \cite{QIC}, \cite{Hube}, \cite{Migdal}, \cite{Ribe1}). In particular, they may  arise as  ground states of  Hamiltonians in nuclear and trapped ion and atom physics such as the Lipkin-Meshkov-Glick Hamiltonian \cite{109} or as examples of violation of Bell's inequalities in quantum nonlocality \cite{primer}. They can also form a set of {\it structured data} in geometric quantum machine learning protocols \cite{Marco2}. Tests of entanglement and separability have been established for such states  (see, e.g., \cite{Miko1}, \cite{Qian}, \cite{Rutko}, \cite{Yu}). 

In this paper, we want to establish yet another property of the symmetric sector and therefore of the CG decomposition of ${\cal H}:=(\mathbb{C}^d)^{\otimes n}$. We look at the orthogonal complement of the symmetric sector in ${\cal H}$ and prove that it does not contain any separable state. It  only contains entangled states. Equivalently, any separable state in ${\cal H}$  must have a non zero component on the symmetric sector. We will prove this as a consequence of a general property of products of elements of a matrix and will establish it   first for pure states and then naturally extend  it to mixed states by convex combination.   For the cases  $n=3$ particles and arbitrary dimension $d$, we   shall obtain a more detailed result,  giving a (nonzero) sharp {\it lower bound} on the size of the component of any separable state on the symmetric sector. A bound for the bipartite case was already known (see, e.g., \cite{EWreview}) and we shall recall such a result also. This will lead us to define a family of {\it entanglement witnesses}  for these systems. We will show with an example that this class of entanglement witnesses can detect entangled states that cannot be detected using the standard partial transposition test.

The paper is organized as follows. In section \ref{SSS},  we describe the states in the symmetric sector and in particular  the  {\it Dicke states} which form an orthonormal  basis of this subspace of  ${\cal H}:=(\mathbb{C}^d)^{\otimes n}$. 
In section \ref{ATL} we prove a technical linear algebra lemma which is the key step to prove the above described property of the symmetric sector. This property  is proved in section \ref{MR}. In section \ref{N23}, we obtain  sharp lower bounds for the component of any separable state on the symmetric sector for the cases of $n=2$ and $n=3$. Bounds for $n>3$ can be obtained numerically after setting up an appropriate multivariable optimization problem and we give an example in the appendix for the case $n=4$, $d=2$. The whole treatment is carried  out for {\it pure} states, {but the extension to mixed states is natural by convex combination and it is discussed in section \ref{MRex}.} In section \ref{EW3} we interpret the above bounds as specifying {\it entanglement witnesses}. After setting up the appropriate definitions, we use an example to show that such a class of entanglement witnesses are not {\it decomposable}  and detect { states that cannot be detected using the (multipartite version) of the standard {\it Peres-Horodecki partial transposition criterion} (see, e.g., \cite{90}), that is,  Positive Partial Transposed (PPT) entangled states}. 

\section{The Symmetric Sector and Dicke States}\label{SSS}

When dealing with systems in quantum information theory, it is customary to work with an orthonormal  basis for $\mathbb{C}^d$, often referred to 
as the {\it computational basis},  given by 
$\{ |0\rangle, |1 \rangle,...,|d-1\rangle \}$. Accordingly, an orthonormal  basis of 
${\cal H}:=(\mathbb{C}^d)^{\otimes n}$ is  given by $|j_1 \, j_2 \, \cdots j_n\rangle$, with $j_1,j_2,...,j_n \in \{0,1,...,d-1\}$. The {\it symmetric sector} is 
the subspace of states invariant under any element of the permutation group $S_n$. 

Consider now the set $\Delta_{d,n}$ of $d$-tuples $[k_0,k_1,...,k_{d-1}]$ with natural  numbers $k_j \geq 0$ and $k_0+k_1+ \cdots + k_{d-1}=n$. We have  
that $|\Delta_{d,n}|=\begin{pmatrix} n+d-1 \cr d-1\end{pmatrix}$. For any element $[k_0,...,k_{d-1}]$ in $\Delta_{d,n}$ denote by $G_{[k_0,...,k_{d-1}]}$ the set of $n$-ples, $(j_1,,...,j_n)$ with $k_0$ $0$'s, $k_1$, $1$'s,...,$k_{d-1}$, $d-1$'s. For example set $d=3$, $n=4$,  and consider 
$[1,0,3] \in \Delta_{3,4}$. Then every $4$-ple with one  `$0$'  and three `$2$'s  and zero $1$'s belongs to $G_{[1,0,3]}$. For example $(0,2,2,2)$, but also $(2,0,2,2)$ etc. The set $G_{[k_0,k_1,...,k_{d-1}]}$ contains $\left| G_{[k_0,k_1,...,k_{d-1}]} \right|=\frac{n!}{k_0!k_1!\cdots k_{d-1}!}$ elements. To any $d$-tuple 
$[k_0,k_1,...,k_{d-1}] \in \Delta_{d,n}$ we associate a {\it Dicke state}, $|\phi_{[k_0,k_1,...,k_{d-1}]} \rangle $,  defined as
\begin{equation}\label{definitionDicke}
|\phi_{[k_0,k_1,...,k_{d-1}]}\rangle := \frac{1}{\sqrt{|G_{[k_0,k_1,...,k_{d-1}]}|}} \sum_{(j_1,...,j_n) \in G_{[k_0,k_1,...,k_{d-1}]}} |j_1,...,j_n \rangle=
\end{equation}
$$
 \sqrt{\frac{k_0! k_1! \cdots k_{d-1}!}{n!} } \sum_{(j_1,...,j_n) \in G_{[k_0,k_1,...,k_{d-1}]}} |j_1,...,j_n \rangle. 
$$ 
For example, for $d=3$, $n=3$, the state $|\phi_{[1,0,2]}\rangle $ is 
$
|\phi_{[1,0,2]} \rangle :=\frac{1}{\sqrt{3}} \left( |0 \,  2 \, 2 \rangle + |2 \,  0 \, 2 \rangle + |2 \, 2 \, 0 \rangle \right).  
$
Dicke states form an orthonormal  basis of the symmetric sector, which has therefore dimension $|\Delta_{d,n}|=\begin{pmatrix} n+d -1 \cr d-1 \end{pmatrix} $. This is the basis we shall refer to in the following.

In terms of the Dicke states, we can write the {\it orthogonal projection}, $\Pi$,  onto the symmetric sector  as,  
\begin{equation}\label{Proiezione}
\Pi:=\sum_{[k_0,k_1,...,k_{d-1}] \in \Delta_{d,n}} |\phi_{[k_0,k_1,...,k_{d-1}]} \rangle \langle \phi_{[k_0,k_1,...,k_{d-1}]} |. 
\end{equation}
{The projection can also be written in terms of the permutations in the symmetric group $S_n$, that is,  as 
\begin{equation}\label{Proiezione2}
\Pi=\frac{1}{n!} \sum_{S \in S_n} S=\frac{1}{n!} \left(  \sum_{S \in S_n} \, \sum_{(j_1,...,j_n) \in \{0,1,...,d-1\}} |S(j_1,...,j_n)\rangle \langle j_1,...,j_n| \right).  
\end{equation}
}

\section{A Technical Lemma}\label{ATL}

Consider now a (real or complex) $n \times d$ matrix $F_{n \times d}$ with entries $f_{l,m}$, $l=1,2,...,n$, $m=0,1,...,d-1$, and consider for  $F_{n \times d}$   the following property. 

\vspace{0.25cm}

{\bf Property A} For any element $[k_0,k_1,...,k_{d-1}] \in \Delta_{d,n}$ we have 
\begin{equation}\label{PropeA}
\sum_{(j_1,j_2,...,j_n) \in G_{[k_0,k_1,...,k_{d-1}]}} f_{1,j_1} f_{2,j_2}\cdots f_{n,j_n}= 0. 
\end{equation}

If a matrix $F_{n \times d}$ has a zero row, then, automatically it satisfies Property A, since the products that appear in the sums in (\ref{PropeA}) always have one zero element, the one corresponding to the zero row. The following lemma shows that the converse fact is also true, that is, if a matrix satisfies Property A then one of its rows must be zero. Notice also (something we will use)  that both properties are independent of the ordering of the rows, or, in other terms, they are invariant 
if we replace $F_{n \times d}$ by $P F_{n \times d}$ where $P$ is a permutation matrix, since this would only change the order in the products in (\ref{PropeA}). 

\begin{lemma}\label{CombLem}
Assume a matrix $F_{n \times d}$ has Property A. Then $F_{n \times d}$ has  one zero row.
\end{lemma}

\begin{proof}
The proof is by induction on $d$. For $d=1$ the statement is true. In this case $\Delta_{1,n}$ has only one element $[n]$ and $G_{[n]}$ has only the element $(0,0,...,0)$ with $n$ zeros. Property A has only one condition 
\begin{equation}\label{all0}
f_{1,0}f_{2,0}\cdots f_{n,0}=0, 
\end{equation}
which implies that one among $\{f_{1,0},f_{2,0},...,f_{n,0}\}$ is zero. Since each one of these entries is the only entry in a row (there is only one column), the base step is proven.

We now assume the Lemma to be valid for $d-1$. In order to prove it for general $d$ we prove the following {\bf Claim:}  If $F_{n \times d}$ has no zero row we must have 
$f_{1,0}=f_{2,0}=\cdots=f_{n,0}=0$, that is, the first column is zero. To prove this, we prove by induction on $l$ that $f_{l,0}=0$, for $l=1,...,n$. We start by using property (\ref{PropeA}) for $[n,0,...,0]$ which gives, with the only element $(0,0,...,0) \in G_{[n,0,...,0]}$,   formula (\ref{all0}). Therefore one among $\{f_{1,0},f_{2,0},...,f_{n,0}\}$ is zero and we choose $f_{1,0}$ since we can always reorder the rows. 
This gives the base step. Assume now that $f_{1,0}=f_{2,0}=\cdots=f_{l,0}=0$, for $l$,  $n>l\geq 1$. We want to show that $f_{l+1,0}=0$. Consider the conditions (\ref{PropeA}) for an element of $\Delta_{d,n}$ of the form  $[n-l, k_1,k_2,...,k_{d-1}]$ with $k_1+k_2+\cdots+ k_{d-1}=l$. Every multi-index $(j_1,j_2,...,j_n)$ in the sum has $n-l$ zeros. All the elements that have a $0$ in the first $l$ positions give a zero element in the sum, since by the inductive assumption, we have $f_{1,0}=f_{2,0}=\cdots=f_{l,0}=0$. The only possibly nonzero terms are the ones where  the zeros appear all in the last $n-l$ position. Therefore the sum can be written as 
$$
\sum_{(j_1,j_2,...,j_n) \in G_{[n-l,k_1,...,k_{d-1}]}} f_{1,j_1} f_{2,j_2}\cdots f_{l,j_l} f_{l+1,j_{l+1}} \cdots f_{n,j_n}=$$
$$\sum_{(j_1,j_2,...,j_l,0,...,0) \in G_{[n-l,k_1,...,k_{d-1}]}} f_{1,j_1} f_{2,j_2}\cdots f_{l,j_l} f_{l+1,0} \cdots f_{n,0}=
$$
$$f_{l+1,0} \cdots f_{n,0}\left( \sum_{(j_1,...,j_l)\in G_{[k_1,...,k_{d-1}]}} f_{1,j_1}\cdots f_{l,j_l} \right)=0. 
$$
Notice that, in the second to last element of the above string of equalities, $[k_1,...,k_{d-1}]$ is an element of $\Delta_{d-1,l}$. Since the above equalities have to hold for any element of  $\Delta_{d-1,l}$, either $f_{l+1,0} \cdots f_{n,0}=0$ or $\sum_{(j_1,...,j_l)\in G_{[k_1,...,k_{d-1}]}} f_{1,j_1}\cdots f_{l,j_l} =0$ for any $[k_1,...,k_{d-1}] \in \Delta_{d-1,l}$. However, in the second case, this would mean that the matrix obtained from $F_{n \times d}$ selecting the first $l$ rows and neglecting the column $0$ has, by the inductive assumption applied to $d-1$,  a zero row. This would mean, since $f_{1,0}=f_{2,0}=\cdots=f_{l,0}=0$,  that the original matrix has a zero row, which we have excluded (and would conclude the proof of the lemma). Therefore we have 
\begin{equation}\label{firstcolumn}
f_{1,0}=f_{2,0}=\cdots=f_{n,0}=0. 
\end{equation}
Now consider the equations (\ref{PropeA}) corresponding to elements $[0,k_1,...,k_{d-1}]$ in $\Delta_{d,n}$, that is,  the elements $(j_1,...,j_n) \in G_{[0,k_1,...,k_{d-1}]}$ are $n$-ples which have no $0$'s. We have, specializing (\ref{PropeA}),  
$$
0=\sum_{(j_1,...,j_n)\in G_{[0,k_1,...,k_{d-1}]}} f_{1,j_1}\cdots f_{n,j_n}=\sum_{(j_1,...,j_n)\in G_{[k_1,...,k_{d-1}]}}f_{1,j_1}\cdots f_{n,j_n}. 
$$
In the last sum, $[k_1,...,k_{d-1}]$ is an arbitrary element of $\Delta_{d-1,n}$. Therefore, by the inductive assumption, the matrix $\{ f_{l,m}\}$, $l=1,...,n$, $m=1,...,d-1$ has one zero row. This matrix is given by the last $d-1$ columns of $F_{n \times d}$. Since the first column is also zero from (\ref{firstcolumn}), $F_{n\times d}$ has one zero row, and the lemma  is proved. 

\end{proof}


\section{Entanglement of States that are Zero on the Symmetric Sector}\label{MR}

Assume now that we have a separable state $|\psi\rangle$ in ${\cal H}$, that is, 
$|\psi \rangle:=|\psi_1\rangle \otimes |\psi_2 \rangle \otimes \cdots \otimes |\psi_n\rangle$. Writing each $|\psi_m \rangle$, $m=1,...,n$, in the computational basis as $|\psi_m \rangle :=\sum_{j_m=0}^{d-1} f_{m,j_m}|j_m\rangle$, we can write 
\begin{equation}\label{psi}
|\psi \rangle = \sum_{j_1,j_2,...,j_n=0}^{d-1} f_{1,j_1} f_{2,j_2} \cdots f_{n,j_n} |j_1 \rangle \otimes |j_2 \rangle \otimes \cdots \otimes |j_n\rangle.  
\end{equation}  
Now consider the inner product of  $|\psi\rangle $ with a Dicke state $|\phi_{[k_0,k_1,...,k_{d-1}]} \rangle$ in (\ref{definitionDicke}).  We have 
\begin{equation}\label{toBeUs}
\sqrt{|G_{[k_0,k_1,...,k_{d-1}]}|} \langle \phi_{[k_0, k_1,...,k_{d-1}]} |\psi \rangle=
\end{equation}
$$ \left( \sum_{(l_1,...,l_n)\in G_{[k_0,k_1,...,k_{d-1}]}} \langle l_1| \otimes \langle l_2| \otimes \cdots \otimes \langle l_n|  \right)
\left(  \sum_{j_1,...,j_n=0}^{d-1} f_{1,j_1}  \cdots f_{n,j_n} |j_1 \rangle \otimes |j_2 \rangle \otimes \cdots \otimes |j_n\rangle.\right)= 
$$
$$
\sum_{(l_1,...,l_n)\in G_{[k_0,k_1,...,k_{d-1}]}} \sum_{j_1,...,j_n=0}^{d-1}
 f_{1,j_1} \cdots f_{n,j_n} \delta_{l_1,j_1} 
 \delta_{l_2,j_2} \cdots \delta_{l_n, j_n}= 
 $$
 $$
 \sum_{(j_1,...,j_n) \in G_{[k_0,k_1,...,k_{d-1}]}}f_{1,j_1} \cdots f_{n,j_n}. 
$$
Using this, if the separable state $|\psi \rangle $ has zero component in the symmetric sector, since all the inner products $\langle \phi_{[k_0, k_1,...,k_{d-1}]} |\psi \rangle$ are zero,  the coefficients 
$f_{1,j_1}, f_{2,j_2},...,f_{n,j_n}$, $j_1,j_2,...,j_n=0,1,...,d-1$ must  satisfy Property A with (\ref{PropeA}). Therefore, we can apply Lemma \ref{CombLem} to conclude that for at least one $k=1,...,n$, $f_{k,j_k}=0$, for all $j_k=0,1,...,d-1$. This means that one of the factors,  $|\psi_k\rangle$, $k=1,...,n$, in 
the factorization $|\psi\rangle =|\psi_1 \otimes |\psi_2\rangle \otimes \cdots \otimes |\psi_n\rangle$ is zero. However this is impossible for a meaningful quantum state $ |\psi\rangle $. Therefore, we have proven any one of the equivalent statements in the following theorem. 

\begin{theorem}\label{MainT}

\begin{enumerate}

\item The orthogonal complement of the symmetric sector in ${\cal H}:=(\mathbb{C}^d)^{\otimes n}$ contains no separable state. 

\item Any separable state in ${\cal H}:=(\mathbb{C}^d)^{\otimes n}$ must have a nonzero orthogonal component on the symmetric sector. 

\item Any quantum state that has zero component on the symmetric sector must be entangled. 

\end{enumerate}

\end{theorem}

This result may be suggested by and generalizes the situation in the two qubits case where the symmetric sector is spanned by 
$\{ |0,0\rangle, \frac{1}{\sqrt{2}} (|0,1\rangle+ |1,0\rangle), |1,1\rangle\}$ and its orthogonal complement is the antisymmetric subspace spanned by 
$\frac{1}{\sqrt{2}}(|0,1\rangle -|1,0\rangle)$, an  entangled state. 

Some remarks that will be used in the following sections are in order. The $n \times d$ matrix $F$ with entries $f_{m,j}$, $m=1,...,n$, $j=0,1,...,d-1$, can be used to represent an $n$-ple of states $\{ |\psi_m \rangle \}$, $m=1,...,n$,  which appear in the product $\otimes_{m=1}^n |\psi_m\rangle$ for a separable state. Every row represents the components of each of the factor states  in the computational basis. We shall refer to this matrix as the $F$-matrix associated to the separable state $\otimes_{m=1}^n |\psi_m \rangle$.   Rows can be modified by multiplying all the elements by a phase factor $e^{i\phi}$, without modifying the physical state. The component of the product state $|\psi\rangle$ onto the Dicke state $|\phi_{[k_0,k_1,...,k_{d-1}]}\rangle$ is given, according to formula (\ref{toBeUs}), 
by 
\begin{equation}\label{accor}
\langle \phi_{[k_0,k_1,...,k_{d-1}]} | \psi\rangle=\sqrt{\frac{k_0! k_1! \cdots k_{d-1}!}{n!}}  \sum_{(j_1,...,j_n) \in G_{[k_0,k_1,...,k_{d-1}]}}f_{1,j_1} f_{2,j_2}\cdots f_{n,j_n}.  
\end{equation}
Thus the square of the norm of the  orthogonal component onto the symmetric sector,  $\| |\psi \rangle \|_{SS}^2 $,  is
\begin{equation}\label{normsquare}
\| |\psi \rangle \|_{SS}^2 =\sum_{[k_0,k_1,...,k_{d-1}]\in \Delta_{d,n}} \frac{k_0! k_1! \cdots k_{d-1}!}{n!}\left| \sum_{(j_1,j_2,...,j_n) \in G_{[k_0,k_1,...,k_{d-1}]}}f_{1,j_1} f_{2,j_2}\cdots f_{n,j_n}
\right|^2.
\end{equation}
{The quantity $ \| |\psi \rangle \|_{SS}^2$ can also be written in terms of the expression (\ref{Proiezione2}) for the projection $\Pi$, i.e., 
\begin{equation}\label{anotherway}
\| |\psi \rangle \|_{SS}^2=\langle \psi | \Pi | \psi \rangle=\frac{1}{n!} \sum_{S \in S_n} \langle \psi | S | \psi \rangle. 
\end{equation}
To compute $\langle \psi | S | \psi \rangle$ in the case where $|\psi\rangle:=|\psi_1 \rangle \otimes |\psi_2 \otimes \cdots \otimes |\psi_n\rangle$ we can use the fact that every permutation $S$ is the product of disjoint cycles, $C_1,C_2,...,C_m$ so that, 
$\langle \psi | S | \psi \rangle=\prod_{j=1}^m \langle \psi | C_j |\psi \rangle$, where if $C$ is the cycle $C=(i_1 \, i_2\, \cdots \, i_k)$, $\langle \psi | C |\psi \rangle=\langle \psi_{i_1}|\psi_k \rangle \langle \psi_{i_2}|\psi_{i_1} \rangle \langle \psi_{i_3}|\psi_{i_2} \rangle \cdots \langle \psi_k | \psi_{i_{k-1}}\rangle$. In the following section we shall calculate the minimum value for $\| |\psi \rangle\|_{SS}^2$ for $n=2,3$ and we  shall use mostly the method based on  (\ref{normsquare}). In the appendix we shall consider the case $n=4$, $d=2$ and we shall use the expression (\ref{anotherway}). This makes a connection with the structure theory of the symmetric group $S_n$. 
}
If $X \in U(d)$ and we define 
$Y:=X \otimes X \otimes \cdots \otimes X$, the tensor product of $n$ equal factors, we have  $Y^\dagger \Pi Y=\Pi$.\footnote{To see this, just apply left hand side and right hand side to an arbitrary vector in $(\mathbb{C}^d)^{\otimes n}$ written as the sum of its components in the symmetric sector and in the orthogonal complement and use the fact that both these spaces are invariant under $Y$ and $\Pi$ is the identity on the symmetric sector and zero on its orthogonal complement.} If $|\psi \rangle$ is separable, we can choose $X$ to make the first factor equal to $|0\rangle$, and then the second factor of the form $\alpha |0 \rangle + \beta |1\rangle$ (with $\alpha$ and $\beta$ real), and then the third factor of the form $\alpha |0\rangle + \beta |1\rangle + \gamma |2 \rangle$ and  (possibly) so on, without loss of generality.

\section{Lower bounds for $\| | \psi \rangle \|_{SS}^2$ for separable $|\psi \rangle$}\label{N23}

\subsection{Case $n=2$} 
The bipartite case $n=2$ is treated as an example in \cite{EWreview}   (Example 3.1).  In this case $\Pi$ in (\ref{Proiezione2}) is $\Pi=\frac{1}{2} \left( {\bf 1} + (1 \,2) \right)$, where $(1 \,2)$ is the {\it SWAP operator} permuting the first and second factor. If $|\psi\rangle$ is separable, $|\psi \rangle =|\psi_1 \rangle \otimes |\psi_2 \rangle$, then 
$$
\| | \psi \rangle \|_{SS}^2= \langle \psi | \Pi |\psi \rangle=\frac{1}{2} +\frac{1}{2} |\langle \psi_1 | \psi_2 \rangle |^2 \geq \frac{1}{2}, 
$$
where equality is achieved when $|\psi_1\rangle $ and $|\psi_2\rangle$ are orthogonal to each other. 

\subsection{Case $n=3$; $d=2$} 

To compute the minimum $\||\psi \rangle \|_{SS}^2$ for the case $n=3$, $d=2$, we shall use the expression (\ref{normsquare}) for $\| | \psi \rangle \|_{SS}^2$. 
According to the discussion at the end of the previous section, we can apply a common unitary on the qubits and collect the phase factors, so that the resulting $F-$matrix  has, without loss of generality, the form, 
\begin{equation}\label{PPOL}
F=\begin{pmatrix} 1 & 0 \cr 
\cos(\theta) & \sin(\theta) \cr 
\cos(\eta) & \sin(\eta) e^{i \alpha}
\end{pmatrix}.
\end{equation}
Following (\ref{accor}) The components along the Dicke states in this case are 
$$
\langle \phi_{[3,0]} |\psi \rangle= f_{1,0}f_{2,0} f_{3,0} =\cos(\theta) \cos(\eta), 
$$
$$\langle \phi_{[2,1]}|\psi \rangle=\frac{1}{\sqrt{3}} \left( f_{1,1} f_{2,0} f_{3,0}+  
f_{1,0} f_{2,1} f_{3,0}+f_{1,0} f_{2,0} f_{3,1} \right)=\frac{1}{\sqrt{3}} \left( \sin(\theta) \cos((\eta)+ \cos(\theta) \sin(\eta) e^{i\alpha} \right), 
$$
$$
 \langle \phi_{[1,2]}|\psi \rangle=\frac{1}{\sqrt{3}} \left( f_{1,1} f_{2,1} f_{3,0}+  
f_{1,0} f_{2,1} f_{3,1}+f_{1,1} f_{2,0} f_{3,1} \right)=\frac{1}{\sqrt{3}} \sin (\theta) \\sin(\eta) e^{i \alpha}, 
$$
$$
\langle \phi_{[0,3]} |\psi \rangle= f_{1,1}f_{2,1} f_{3,1}=0. 
$$
Therefore we have 
\begin{equation}\label{Ex5}
\| |\psi \rangle \|_{SS}^2=\cos^2(\theta)\cos^2(\eta)+
\end{equation}
$$\frac{1}{3} \left( \sin^2(\theta) \cos^2(\eta) + \cos^2(\theta) \sin^2(\eta) + 
2 \sin(\theta) \cos(\eta) \cos(\theta) \sin(\eta) \cos(\alpha) \right)+ \frac{1}{3} \sin^2(\theta) \sin^2(\eta) = 
$$
 $$
 \frac{1}{3}+ \frac{2}{3} \cos^2(\theta) \cos^2(\eta) +\frac{1}{6} \sin(2 \theta) \sin(2 \eta) \cos(\alpha)= 
 $$
 $$\frac{1}{3}+ \frac{2}{3} \left( \frac{1+\cos(2\theta)}{2}\right) \left( \frac{1+\cos(2\eta)}{2}\right) +\frac{1}{6} \sin(2 \theta) \sin(2 \eta) \cos(\alpha)=
 $$
 $$
 \frac{1}{2}+ \frac{1}{6} \left( \cos(2 \eta) + (1+\cos(2\eta))\cos(2 \theta) +  \sin(2 \theta) \sin(2 \eta) \cos(\alpha) \right). 
 $$
 Now fix $\eta$ and $\alpha$ and minimize this expression over $\theta$, using the fact that the minimum of $a \cos(x)+b\sin(x)$ is $-\sqrt{a^2+b^2}$. This gives 
 $$
 \| |\psi \rangle \|_{SS}^2 \geq \frac{1}{2} + \frac{1}{6} \left( \cos(2 \eta)-\sqrt{(1+\cos(2\eta))^2+ \sin^2(2 \eta) \cos^2(\alpha)} \right) \geq 
 $$
 $$\frac{1}{2} + \frac{1}{6} \left( \cos(2 \eta)-\sqrt{(1+\cos(2\eta))^2+ \sin^2(2 \eta) } \right)= 
 $$
$$
\frac{1}{2} +\frac{1}{6} \left( \cos(2 \eta)-\sqrt{2+2 \cos(2\eta)} \right)= \frac{1}{2} + \frac{1}{6} \left( 2 \cos^2(\eta) -1 -2|\cos(\eta)|\right)= 
\frac{1}{3}+ \frac{1}{3} \left( \cos^2(\eta) - |\cos(\eta) \right) \geq \frac{1}{4}, 
$$
where we used the fact that the minimum of $x^2-x$ in $[0,1]$ is $-\frac{1}{4}$. This bound of $\frac{1}{4}$ is achieved by setting in (\ref{Ex5}) $\alpha=0$, $\theta=\frac{2 \pi}{3}$, $\eta=\frac{\pi}{3}$.

\subsection{Case $n=3$; $d\geq 3$} 
In the case $n=3$ and $d \geq 3$, we use again formulas (\ref{accor}) and (\ref{normsquare}). By applying a common 
unitary following the discussion at the end of the previous section, and collecting a phase factor,  we can assume the $F$-matrix, without loss of generality, of the form  
\begin{equation}\label{newF}
F=\begin{pmatrix}
1 & 0 & 0 & 0 & \cdots & 0 \cr 
\cos(\theta) & \sin (\theta) & 0 & 0 & \cdots & 0 \cr 
\cos(\eta) & \sin(\eta) \cos(\gamma) e^{i \rho} & \sin (\eta) \sin(\gamma) & 0 & \cdots & 0 
\end{pmatrix}. 
\end{equation}
Consider now a $d$-tuple $[k_0,k_1,k_2,...,k_{d-1}]$ with one $k_l$ for $l \geq 3$ different from zero. Then every triple $(j_1,j_2,j_3)$ in the sum in (\ref{accor}) would have an element equal to $l$, which would make every term in the sum equal to zero  from the form of the matrix $F$ in (\ref{newF}). Similarly, every $d$-tuple such that $k_0=0$ will be such that $(j_1,j_2,j_3) \in G_{[k_0, k_1,...,k_{d-1}]}$ implies $j_1 \not=0$. Therefore, all the terms in the sum (\ref{accor}) will be zero, from the $F$-matrix in (\ref{newF}). Therefore the only components on Dicke states we have to consider are  
$$
\langle \phi_{[3,0,0,0,...,0]} |\psi \rangle= f_{1,0}f_{2,0} f_{3,0} =\cos(\theta) \cos(\eta), 
$$
$$
\langle \phi_{[2,1,0,0,...,0]} |\psi \rangle= \frac{1}{\sqrt{3}} \left(f_{1,0}f_{2,0} f_{3,1}+f_{1,0}f_{2,1} f_{3,0}+f_{1,1}f_{2,0} f_{3,0} \right) =\frac{1}{\sqrt{3}} \left( \cos(\theta) \sin(\eta) \cos(\gamma) e^{i \rho} +\sin(\theta) \cos(\eta)\right), 
$$
$$
\langle \phi_{[2,0,1,0,...,0]} |\psi \rangle= \frac{1}{\sqrt{3}} \left(f_{1,0}f_{2,0} f_{3,2}+f_{1,0}f_{2,2} f_{3,0}+f_{1,2}f_{2,0} f_{3,0} \right) = \frac{1}{\sqrt{3}} \cos(\theta) \sin(\eta) \sin(\gamma), 
$$
$$
\langle \phi_{[1,2,0,0,...,0]} |\psi \rangle= \frac{1}{\sqrt{3}} \left(f_{1,0}f_{2,1} f_{3,1}+f_{1,1}f_{2,1} f_{3,0}+f_{1,1}f_{2,0} f_{3,1} \right) = 
\frac{1}{\sqrt{3}} \sin(\theta) \sin(\eta) \cos(\gamma) e^{i \rho}, 
$$
$$
\langle \phi_{[1,1,1,0,...,0]} |\psi \rangle= \frac{1}{\sqrt{6}} \left(f_{1,0}f_{2,1} f_{3,2}+f_{1,0}f_{2,2} f_{3,1}+f_{1,1}f_{2,0} f_{3,2}
+f_{1,1}f_{2,2} f_{3,0}+ f_{1,2}f_{2,1} f_{3,0}+ f_{1,2}f_{2,0} f_{3,1}\right) =
$$
$$ \frac{1}{\sqrt{6}} \sin (\theta) \sin(\eta) \sin(\gamma). 
$$
We have 
$$
\| |\psi \rangle \|_{SS}^2=\cos^2(\theta) \cos^2(\eta)+\frac{1}{3} \cos^2(\theta) \sin^2(\eta) \cos^2(\gamma)+ \frac{1}{3} \sin^2(\theta) \cos^2(\eta)+ \frac{1}{6} \sin(2 \theta) \sin(2 \eta) \cos(\gamma) \cos(\rho)+$$
$$
\frac{1}{3} \cos^2(\theta) \sin^2(\eta) \sin^2(\gamma) 
+ \frac{1}{3} \sin^2(\theta)\sin^2(\eta)\cos^2(\gamma)+\frac{1}{6} \sin^2(\theta) \sin^2(\eta) \sin^2(\gamma)=
$$

$$
\cos^2(\theta) \cos^2(\eta)+\frac{1}{3} \cos^2(\theta) \sin^2(\eta)+\frac{1}{3} \sin^2(\theta)\cos^2(\eta) +
$$
$$
\frac{1}{6} \sin(2 \theta)\sin(2 \eta) \cos(\gamma) \cos(\rho) +\frac{1}{6} \sin^2(\theta) \sin^2(\eta) \cos^2(\gamma) + \frac{1}{6} \sin^2(\theta) \sin^2(\eta) =
$$
$$
\frac{1}{3}\cos^2(\theta) \cos^2(\eta)+\frac{1}{3} \cos^2(\theta) +\frac{1}{3} \cos^2(\eta) +
$$
$$
 \frac{1}{6} \sin^2(\theta) \sin^2(\eta) + \frac{1}{6}\left[  \sin(2 \theta)\sin(2 \eta) \cos(\gamma) \cos(\rho) + \sin^2(\theta) \sin^2(\eta) \cos^2(\gamma) \right]. 
$$

For a fixed $\theta$, $\eta$ and $\rho$, let us minimize the quantity in square brackets with respect to $\gamma$. In the simple case where $\sin(\theta)=0$ and/or $\sin(\eta)=0$, the quantity in square brackets is identically zero. Let us assume, without loss of generality that $\sin(\theta)=0$ since $\theta$ and $\eta$ are intercheangable. The function becomes equal to $\frac{1}{3}+\frac{2}{3} \cos^2(\eta)\geq \frac{1}{3}$ . 
Let us now consider the case where $\sin(\theta) \not=0$ and $\sin(\eta) \not=0$. A calculus exercise shows that the minimum of a function $f(x)=bx^2+ax$ with $b >0$ is $-\frac{a^2}{4b}$ if $\frac{|a|}{2b} < 1$ and $b-|a|$ if $\frac{|a|}{2b} \geq 1$. Applying this, to the function above with $b=\sin^2(\theta) \cos^2(\theta)$, $a=\sin(2\theta) \sin(2\eta) \cos(\rho)$, we have to consider two situations: 

\begin{enumerate}

\item 
$$
\frac{|a|}{2b} =\frac{|sin(2\theta) \sin(2\eta) \cos(\rho)|}{2 \sin^2(\theta) \sin^2(\eta)}=\frac{|2\cos(\theta) \cos(\eta) \cos(\rho)|}{|\sin(\theta) \sin(\eta)|} < 1$$ 
In this case, using the expression $-\frac{a^2}{4b}$ for the minimum of the function in parenthesis, we get, 
$$
\|\psi\|_{SS}^2 \geq \frac{1}{3} \cos^2(\theta) \cos^2(\eta)+\frac{1}{3} \cos^2(\eta)+\frac{1}{3} \cos^2(\theta) +\frac{1}{6} \sin^2(\theta) \sin^2(\eta) -\frac{1}{6} \left( \frac{\sin^2(2 \theta) \sin^2(2 \eta) \cos^2(\rho)}{4\sin^2(\theta) \sin^2(\eta)} \right) \geq 
$$
$$
-\frac{1}{3} \cos^2(\theta) \cos^2(\eta)+ \frac{1}{3} \cos^2(\eta) +\frac{1}{3} \cos^2(\theta)+\frac{1}{6} \sin^2(\theta) \sin^2(\theta)=
$$
$$
\frac{1}{6}+\frac{1}{6} \left(\cos^2(\eta) +\cos^2(\theta)-\cos^2(\theta) \cos^2(\eta) \right) \geq \frac{1}{6}. 
$$
Here, in the last inequality, we used the fact that the minimum of the function $f(x,y)=x+y-xy$ in the square $0 \leq x \leq 1$, $0 \leq y \leq 1$, is zero. Such a minimum is achieved with $x=y=0$ and therefore $\theta=\frac{\pi}{2}$, $\eta=\frac{\pi}{2}$. Notice that this is also consistent with the assumption $\frac{|a|}{2b}=\frac{|2\cos(\theta) \cos(\eta) \cos(\rho)|}{|\sin(\theta) \sin(\eta)|} <1$. 

\item 
\begin{equation}\label{latertoBus8}
\frac{|a|}{2b} =\frac{|sin(2\theta) \sin(2\eta) \cos(\rho)|}{2 \sin^2(\theta) \sin^2(\eta)}=\frac{|2\cos(\theta) \cos(\eta) \cos(\rho)|}{|\sin(\theta) \sin(\eta)|} \geq 1
\end{equation}
We have, using the value $b-|a|=\sin^2(\theta) \sin^2(\eta) -|\sin(2 \theta) \sin(2 \eta)\cos(\rho)|$ for the minimum, 
$$
\|\psi_{SS}\|^2 \geq \frac{1}{3} \cos^2(\theta) \cos^2(\eta) + \frac{1}{3} \cos^2(\eta)+ \frac{1}{3} \cos^2(\theta) + \frac{1}{3} \sin^2(\theta) \sin^2(\eta) - \frac{1}{6}|\sin(2\theta) \sin(2 \eta) \cos(\rho)|
$$
$$
\| |\psi \rangle \|_{SS}^2 \geq_{\rho=0} \frac{1}{3} \cos^2(\theta) \cos^2(\eta) + \frac{1}{3} \cos^2(\eta) + \frac{1}{3} \cos^2(\theta) +\frac{1}{3} \sin^2(\theta) \sin^2(\eta) -\frac{1}{6} |\sin(2 \theta) \sin(2 \eta)|. 
$$
Using Young inequality $|sin(2\theta) \sin(2 \eta)|\leq \frac{1}{2} \left( \sin^2(2 \theta) + \sin^2(2 \eta) \right)$ and trigonometric identities we get  
$$
\| |\psi \rangle \|_{SS}^2  \geq 
$$
$$
{\frac{1}{3}} \left[ \cos^2(\theta) \cos^2(\eta) + \cos^2(\eta) + \cos^2(\theta)+ \sin^2(\theta) \sin^2(\eta)  -\sin^2(\theta) \cos^2(\theta) - \sin^2(\eta) \cos^2(\eta) \right] =
$$
$$
\frac{1}{3}\left[ (\cos^2(\theta)+ \cos^2(\eta))^2+1 - (\cos^2(\eta) + \cos^2(\theta))\right]. 
$$
We then minimize  this subject to (\ref{latertoBus8}) which for $\rho=0$ gives 
\begin{equation}\label{toBus8}
3 \cos^2(\theta) \cos^2(\eta) \geq 1 - \cos^2(\eta) -\cos^2(\theta),  
\end{equation}
or, with notational changes, we consider the minimization of a function $f(x,y)=$ 

$\frac{1}{3} \left( (x+y)^2+1-x-y \right)$ subject to $0 \leq x,y, \leq 1$ and $3xy \geq 1-x-y$. 
This minimum is achieved for $x=y=\frac{1}{3}$ and it is equal to $\frac{7}{27}$.\footnote{There is no stationary point for the function in the allowed set since stationary points are where $x+y=\frac{1}{2}$ which is incompatible with the given constraints. So the minimum must be on the boundary of the set which is where $x \equiv 1$, $y \equiv 1$ or where $y=\frac{1-x}{3x+1}$ with $0 \leq x \leq 1$. With $x \equiv 1$ ($y\equiv 1$)  the minimum is for $y=0$ ($x=0$) and it is equal to $\frac{1}{3}$. For  $y=\frac{1-x}{3x+1}$ with $0 \leq x \leq 1$, the minimum is obtained for $x=\frac{1}{3}$ (which also gives $y=\frac{1}{3}$) and it is equal to $\frac{7}{27}$ which is smaller than $\frac{1}{3}$.}
\end{enumerate}

In conclusion, comparing all cases, the minimum is $\frac{1}{6}$.

\subsection{Summary} 

We summarize the  lower bounds we  found, which are all sharp,  in the following table 

\begin{center}
\begin{tabular}{ |c | c | c | }
\hline
 $\quad$ & $\quad$ & $\quad$ \\  
$ \| |\psi\rangle \|_{SS}^2 \leq $ & $ d=2 $ & $d \geq 3$  \\ 
 $\quad$ & $\quad$ & $\quad$ \\  
 \hline
  $\quad$ & $\quad$ & $\quad$ \\  
 $n=2$  & ${{\frac{1}{2}}}$  & $\frac{1}{2}$ \\ 
  $\quad$ & $\quad$ & $\quad$ \\  
 \hline 
 $\quad$ & $\quad$ & $\quad$ \\  
 $n=3$  & $\frac{1}{4}$  & $\frac{1}{6}$\\
  $\quad$ & $\quad$ & $\quad$ \\  
 \hline    
\end{tabular}
\end{center}.


\section{Extension to Mixed States}\label{MRex}

The extension of the  result of Theorem \ref{MainT}  to the case of mixed states and in a density matrix description of quantum states comes naturally  once we recall the appropriate definitions. Recall that a general (possibly mixed) quantum state can be described by a density matrix operator as a statistical mixture of pure states, i.e., (see, e.g., \cite{Sakurai}) 
\begin{equation}\label{statmix}
\rho:=\sum_s \lambda_s |\psi_s \rangle \langle \psi_s|, \quad \lambda_s > 0, \quad  \sum_s \lambda_s=1.  
\end{equation}
In this context,  $\rho$ is called {\it separable} if it is a statistical mixture of product states, that can be assumed, without loss of generality,  to be pure,  i.e., 
\begin{equation}\label{sepstate}
\rho=\sum_j {\omega_j} |\psi_{1,j} \rangle \langle \psi_{1,j}|  \otimes  |\psi_{2,j} \rangle \langle \psi_{2,j}|   \otimes \cdots \otimes  |\psi_{n,j} \rangle \langle \psi_{n,j}|=
\end{equation}
$$
\sum_j {\omega_j} \left( |\psi_{1,j} \rangle \otimes |\psi_{2,j} \rangle \otimes \cdots \otimes |\psi_{n,j} \rangle \right) 
\left( \langle \psi_{1,j} | \otimes \langle \psi_{2,j} | \otimes \cdots \otimes \langle \psi_{n,j} | \right) 
$$
with $\sum_j \omega_j=1,$ $\omega_j >0$. 
States $\rho$ that cannot be written as 
in (\ref{sepstate}) are called {\it entangled}. With these 
definitions, the extension of Theorem \ref{MainT} to general 
mixed states reads as follows. 
\begin{theorem}\label{MainTDM}
Assume for every state $|\phi\rangle$ in the symmetric sector $\rho |\phi\rangle=0$, or, equivalently $Tr(\Pi \rho)=0$ with $\Pi$ in (\ref{Proiezione}) (\ref{Proiezione2}). Then $\rho$ is 
an entangled state. 
\end{theorem}  
\begin{proof}
If for every $|\phi\rangle$ in the symmetric sector $\rho |\phi \rangle=0$ and $\rho$ is separable, using (\ref{sepstate}),  we get 
$ \sum_j \omega_j \left| \langle \phi | \psi_{1,j} \rangle \otimes |\psi_{2,j} \rangle \otimes \cdots \otimes |\psi_{n,j} \rangle \right|^2=0$ 
which implies, for every $j$, 
$ \left| \langle \phi | \psi_{1,j} \rangle \otimes |\psi_{2,j} \rangle \otimes \cdots \otimes |\psi_{n,j} \rangle \right|^2=0$, 
which is impossible according to Theorem \ref{MainT}. 
\end{proof} 
Let us consider again $\Pi$,  the orthogonal projection onto the symmetric sector as defined in (\ref{Proiezione}), (\ref{Proiezione2}).  
If we have a lower bound $B$ on the component square  of any separable state  on the symmetric sector, we have that, for any separable state $\rho$, 
$\rho:=\sum_j \omega_j |\psi_j \rangle \langle \psi_j |$, 
$$
Tr(\Pi \rho) = \sum_j \omega_j \langle \psi_j | \Pi |\psi_j \rangle \geq B.  
$$
 Therefore, for any separable state $\rho$ we have (${\bf 1}$ denotes the identity)  
 $$
 Tr \left[ \left( \Pi -B{\bf 1} \right) \rho  \right] =Tr(\Pi \rho) - B \geq 0. 
 $$ 
 Therefore $W:=\Pi-B{\bf 1}$ is an {\it entanglement witness} \cite{EWreview}  for multipartite systems on $(\mathbb{C}^d)^{\otimes n}$ (see, e.g., \cite{NC}) that is an operator $W$ such that $Tr(W \rho) \geq 0$ for any separable state $\rho$. The condition $Tr(W \rho) < 0$ therefore indicates that $\rho$ is an entangled state.

\section{Entanglement witnesses and PPT states}\label{EW3}

It is interesting to know whether the above  class of entanglement witnesses is able to detect entangled states which are not detected by the popular 
{\it Peres-Horodecki Positive Partial Transposition  (PPT) criterion} (see, e.g., \cite{Hororev}). In this section we explore this issue after setting 
up the appropriate definitions. We shall show with an example for a system of three qubits that this is actually the case. 

The  PPT criterion of entanglement was originally proposed for {\it bipartite} quantum systems (of possibly different dimensions) and it turned out to be not only avery easy to use criterion but also a very powerful one. In the bipartite case it  detects  {\it all} entangled states when the product of the dimensions of the two subsystems is less than or equal to $6$  \cite{90}. We discuss it here in the context of a (natural) extension to the multipartite case, assuming, as in the rest of the paper, that all the $n$ subsystems have the same dimension $d$. Let ${\cal I}$ be a subset of ${\cal N}:=\{1,...,n\}$, and denote by $\rho^{\Gamma_{\cal I}}$ the partial transpose of $\rho$ with respect to the factors labeled by the elements in ${\cal I}$.\footnote{That is, expand $\rho$ as a linear combination of tensor products and then take the transpose of only the elements in the products with positions given by the indexes in ${\cal I}$}
The state $\rho$ is called {\it ${\cal I}$-PPT} if  $\rho^{\Gamma_{\cal I}} \geq 0$, that is,  $\rho^{\Gamma_{\cal I}}$ is a state as well (since the trace equal to 1 condition is automatically satisfied). Notice that since 
$(\rho^{\Gamma_{\cal I}})^T=\rho^{\Gamma_{{\cal N}-{\cal I}}} $, $\rho$ is ${\cal I}$-PPT if and only if it is  (${\cal N}-{\cal I}$)-PPT. Every separable state $\rho$ is ${\cal I}$-PPT, no matter what ${\cal I}$ is. Therefore if there exists a set ${\cal I}$ such that $\rho$ is not ${\cal I}$-PPT, $\rho $ is entangled. In our multipartite setting, we shall call PPT states, states that are ${\cal I}$-PPT for any ${\cal I}$. Separable states are PPT. As stated in \cite{EWreview}, `the hard problem in entanglement theory is to detect entangled PPT states'. Notice that,  to check that a state is PPT,  we need to check that it is ${\cal I}$-PPT for any ${\cal I} \subseteq {\cal N}$, that is,  for $2^{n}$ combinations. However  ${\cal I}=\{ \emptyset \}$ and ${\cal I}={\cal N}$ do not need to be checked (since $\rho \geq 0$ and $\rho^T \geq 0$) and ${\cal I}$ and ${\cal N}-{\cal I}$ give the same result. Thus,  we need to check a total of $2^{n-1}-1$ subsets in ${\cal N}$, which reduces to the known  single test in the bipartite case.

We now present a class of PPT states which is detected by the entanglement witnesses found in this paper. We consider three qubits ($n=3$, $d=2$) and the class of {\it generalized Werner states} 
\begin{equation}\label{Werner}
\rho=\frac{p}{4} \Pi+ \frac{1-p}{4}\left( {\bf 1}-\Pi \right),  
\end{equation} 
where $\Pi$ is the projection onto the symmetric sector and $p \in [0, 1]$.    According to the results of section \ref{N23}, $\Pi-\frac{1}{4}{\bf 1}$ is an entanglement witness for the system $n=3$ $d=2$, which applied to the states $\rho$ in (\cite{WernerC}) implies that $\rho$ is entangled for $p< \frac{1}{4}$. When calculating $\rho^{\Gamma_{\cal I}}$ for $\rho$ in (\ref{Werner}) and ${\cal I}$, to check whether $\rho$ is PPT, we only need to check for $|{\cal I}|=1$ since $\rho^{\Gamma_{{\cal N}-{\cal I}}} \geq 0 \leftrightarrow \rho^{\Gamma_{{\cal I}}} \geq 0$. Furthermore, let $P$ any permutation matrix. Since $\rho$ commutes with $P$ we have that for ${\cal I}_1$ and ${\cal I}_2$ of cardinality $1$, 
\begin{equation}\label{riduzione}
\rho^{\Gamma_{{\cal I}_2}}=P\rho^{\Gamma_{{\cal I}_1}}P^T
\end{equation}
 for some permutation $P$.\footnote{In general for ${\cal I}_1$ and ${\cal I}_2$ of the same cardinality we have for an appropriate permutation matrix $P$,  $\rho^{\Gamma_{{\cal I}_2}}=P \left( P^T \rho P \right)^{\Gamma_{{\cal I}_1}} P^T$ which reduces to (\ref{riduzione}), if $\rho$ commutes with $P$ as in our case.} Therefore, we need to check the PPT property only for one set ${\cal I}$ with cardinality $1$ and we take,  without loss of generality,  ${\cal I}=\{ 1\}$, denoting simply by $\Gamma$ the corresponding $\Gamma_{\cal I}$.  
  Applying (\ref{Proiezione2})  this to the case $n=3$, $d=2$ leads to 
  \begin{equation}\label{pi8}
  \Pi=\frac{1}{6} \begin{pmatrix} 6 & 0 & 0 & 0 & 0 & 0 & 0 & 0 \cr 
  0 & 2 &2 & 0 & 2 & 0 & 0 & 0 \cr  
  0 & 2 &2 & 0 & 2 & 0 & 0 & 0 \cr 
  0 & 0 & 0 & 2 & 0 & 2 & 2 & 0 \cr 
  0 & 2 & 2 & 0 & 2 & 0 & 0 & 0 \cr 
  0 & 0 & 0 & 2 & 0 & 2 & 2 & 0 \cr
   0 & 0 & 0 & 2 & 0 & 2 & 2 & 0 \cr 
   0 & 0 & 0 & 0 & 0 & 0 & 0& 6  \end{pmatrix}. 
  \end{equation}
  Partial transposition with respect to the first factor corresponds to swapping the off diagonal $4 \times 4$ blocks of a partitioned matrix 
  $\begin{pmatrix} A& B \cr C & D \end{pmatrix} \rightarrow  \begin{pmatrix} A& C \cr B& D \end{pmatrix} $. Using this and  (\ref{pi8}) in (\ref{Werner}), we obtain 
  \begin{equation}\label{afterPT}
  \rho^{\Gamma}=\frac{p}{24} 
   \begin{pmatrix}6 & 0 & 0 & 0 & 0 & 2 & 2 & 0 \cr
  0 & 2 & 2 & 0 & 0 & 0 & 0& 2 \cr
  0 & 2 & 2 & 0 & 0 & 0 & 0& 2 \cr
  0  & 0 & 0 & 2 & 0 & 0 & 0 & 0 \cr 
   0  & 0 & 0 & 0 & 2 & 0 & 0 & 0 \cr 
  2 & 0 & 0 & 0 & 0 & 2 & 2 & 0 \cr 
    2 & 0 & 0 & 0 & 0 & 2 & 2 & 0 \cr 
    0 & 2 & 2 & 0 & 0 & 0 & 0 &  6\end{pmatrix}
    + \frac{1-p}{24} \begin{pmatrix}0 & 0 & 0 & 0 & 0 & -2 & -2 & 0 \cr 
    0 & 4 & -2 & 0 &   0 & 0 & 0 & {-2} \cr
    0 & -2 & 4 & 0 & 0 & 0 & 0 & -2 \cr  
    0 & 0 & 0 & 4 & 0 & 0 & 0 & 0 \cr 
    0 & 0 & 0 & 0 & 4 & 0 & 0 & 0 \cr
    {-2} & 0 & 0 & 0 & 0 & 4 & {-2} & 0 \cr 
    {-2} & 0 & 0 & 0 & 0 & {-2} & 4 & 0 \cr
    0 & {-2} & {-2} & 0 & 0 & 0 & 0 & 0 \end{pmatrix}. 
   \end{equation}
   Application of the Sylvester criterion (see, e.g., \cite{HJ}) shows that $\rho^\Gamma$ is positive definite if and only  if $p > 0.2$.  Therefore $\rho$ is PPT for every $p > 0.2$. If $0.25 > p > 0.2$, the entanglement witnesses proposed in this paper detect entangled PPT Werner  states in (\ref{Werner}). 
   
   \begin{remark}\label{WernerR}Werner states were introduced in \cite{WernerC} for the $n=2$ case and $d$ arbitrary and considered in \cite{EWreview} (Example 3.2). It is  known that they are separable if and only if they are PPT. Our example shows that this property does not extend to the $n=3$ case. 
     \end{remark}
   
   \begin{remark}\label{deco} In the bipartite case, an entanglement witness $W$ is called {\it decomposable} if there exist positive semidefinite Hermitian matrices $X$ and $Y$ such that $W=X^{\Gamma}+Y$. Such definition naturally extends to the multipartite case by saying that $W$ is decomposable iff there exists a set ${\cal I} \subseteq {\cal N}$ and positive semidefinite $X$ and $Y$ such that  $W=X^{\Gamma_{\cal{I}}}+Y$. 
   Decomposable entanglement witness are not able to detect entangled PPT states because if $\rho^{\Gamma_{\cal I}}\geq 0$,\footnote{Here we use the property of the partial transposition $Tr(A^{\Gamma_{\cal I}}B)=Tr(A B^{\Gamma_{\cal I}})$.}
   $$
   Tr\left( W \rho \right)=Tr\left( ((X^{\Gamma_{\cal I}}+Y)\rho \right)= Tr \left( X^{\Gamma_{\cal I}} \rho\right)+Tr \left(  Y \rho\right)=
   Tr (X \rho^{\Gamma_{\cal I}})+ Tr(Y\rho) \geq 0. 
   $$
   Therefore, the entanglement witnesses proposed here are, in general, not decomposable in the sense of the above definition.\footnote{One can show them to be decomposable in the $n=2$, bipartite case \cite{EWreview}.}
   
   \end{remark}


\section*{Acknowledgement} 

This material is based upon work supported by, or in part by, the U. S. Army Research Laboratory and the U. S. Army Research Office under contract/grant number W911NF2310255.

\appendix

\section{Bound for $\| \langle \psi\|^2_{SS}$ in the case $n=4$, $d=2$} 
We consider in this appendix a numerical calculation of the minimum for $\||\psi \rangle \|_{SS}^2$ for the case of $n=4$ qubits ($d=2$). Besides the interest in finding such a bound, which, as we have seen in the main text, gives an entanglement witness for such systems, the treatment we shall present also gives us an example of using the expression $\| |\psi \rangle \|_{SS}^2=\langle \psi | \Pi |\psi \rangle $ with $\Pi$ in (\ref{Proiezione2}) to set up the function to minimize. This  requires some knowledge of the structure of the symmetric group $S_4$. 

The group $S_4$ contains $24$ elements which are usually described as (products of) elementary cycles:  the identity ${\bf 1}$, six simple transpositions (cycles with length $2$), three double transpositions, $(1\,2)(3\,4)$, $(1\,3)(2\,4)$, $(1\,4)(2\,3)$, eight cycles of length three, six cycles of length four. 
In the following subsections, we consider these classes of permutations and their contributions to $\||\psi\rangle \|_{SS}^2$. For every permutation $S$ in one class the inverse is in the same class,  and this  guarantees that contributions from each class are real. As done in the main text in section \ref{N23}, we can apply a common unitary transformation and collect a common phase factor which leads us to  assume, without loss of generality, that  $|\psi \rangle=|\psi_1\rangle \otimes |\psi_2\rangle \otimes  |\psi_3 \rangle \otimes |\psi_4  \rangle$, with 
$$
|\psi_1 \rangle=|0\rangle, \quad  |\psi_2 \rangle =\cos(\theta_2) |0 \rangle+ \sin(\theta_2)|1\rangle, \quad 
$$
$$
|\psi_3 \rangle =\cos(\theta_3) |0 \rangle+ e^{i \alpha_3}\sin(\theta_3)|1\rangle, \quad  |\psi_4 \rangle= \cos(\theta_4) |0 \rangle+ e^{i \alpha_4}\sin(\theta_4)|1\rangle
$$
In the following we list the contributions to the sum (\ref{anotherway}) given by the various sets of transformations in $S_4$. 

\subsection{Identity}
The identity ${\bf 1}$ gives a contribution $\langle \psi | {\bf 1} |\psi \rangle=1$. 
\subsection{Six transpositions} 
Consider the  six transpositions $(1\,2)$, $(1\,3)$, $(1\,4)$, $(2\,3)$, $(2\,4)$, $(3\,4)$. The transposition $(j\,k)$, gives a contribution $|\langle \psi_j |\psi_k \rangle|^2$. We have therefore the following transpositions  with the associated contributions 
$$
(1\,2) \leftrightarrow \cos^2(\theta_2), \quad (1\,3) \leftrightarrow \cos^2(\theta_3), \quad (1\,4) \leftrightarrow \cos^2(\theta_4),$$
$$
 (2\,3) \leftrightarrow \cos^2(\theta_2) \cos^2(\theta_3) +\sin^2(\theta_2) \sin^2(\theta_3) +\frac{1}{2} \sin(2\theta_2) \sin(2 \theta_3) \cos(\alpha_3), 
$$
$$
(2\,4) \leftrightarrow \cos^2(\theta_2) \cos^2(\theta_4) +\sin^2(\theta_2) \sin^2(\theta_4) +\frac{1}{2} \sin(2\theta_2) \sin(2 \theta_4) \cos(\alpha_4), 
$$
$$
(3\,4) \leftrightarrow \cos^2(\theta_3)\cos^2(\theta_4)+\sin^2(\theta_3) \sin^2(\theta_4) +\frac{1}{2} \sin(2\theta_3)\sin(2\theta_4) \cos(\alpha_3-\alpha_4)
$$ 
\subsection{Three double transpositions} 
For the three double transpositions the contributions to the sum (\ref{anotherway}) are 
$$
(1\,2)(3\,4) \leftrightarrow \cos^2(\theta_2) \cos^2(\theta_3) \cos^2(\theta_4)+ \cos^2(\theta_2) \sin^2(\theta_3) \sin^2(\theta_4) +\frac{\cos^2(\theta_2)}{2} \sin(2 \theta_3) \sin(2 \theta_4) \cos(\alpha_3-\alpha_4), 
$$
$$
(1\,3)(2\,4) \leftrightarrow \cos^2(\theta_3) \cos^2(\theta_2) \cos^2(\theta_4)+ \cos^2(\theta_3) \sin^2(\theta_2) \sin^2(\theta_4) +\frac{\cos^2(\theta_3)}{2} \sin(2 \theta_2) \sin(2 \theta_4) \cos(\alpha_4)
$$
$$
(1\,4)(2\,3) \leftrightarrow \cos^2(\theta_4) \cos^2(\theta_2) \cos^2(\theta_3)+ \cos^2(\theta_4) \sin^2(\theta_2) \sin^2(\theta_3) +\frac{\cos^2(\theta_4)}{2} \sin(2 \theta_2) \sin(2 \theta_3) \cos(\alpha_3), 
$$

\subsection{Eight 3-cycles} 

It is convenient to group the eight three cycles in pairs where one cycle is the inverse of the other. The corresponding contributions to the sum in (\ref{Proiezione2}) are one the complex conjugate of the other and thus are grouped together to give a real value. The four pairs are 
$(1\,2\,3) + (2\,1\,3)$, $(1\,2\,4) + (2\,1\,4)$, $(1\,3\,4) + (3\,1\,4)$, $(2\,3\,4) + (3\,2\,4)$. Let us calculate the contribution of the first pair, $(1\,2\,3) + (2\,1\,3)$, which is, 
$$
\langle \psi_2|\psi_1 \rangle \langle \psi_3|\psi_2 \rangle\langle \psi_1 |\psi_3\rangle +\langle \psi_1|\psi_2 \rangle \langle \psi_3|\psi_1 \rangle\langle \psi_2 |\psi_3 \rangle= 2 \cos^2(\theta_2)\cos^2(\theta_3)+\frac{1}{2} \sin(2 \theta_2) \sin(2 \theta_3) \cos(\alpha_3). 
$$
Analogously, we can obtain the contributions of the other pairs which we list below 
$$
(1\,2\,4) + (2\,1\,4) \leftrightarrow 2 \cos^2(\theta_4) \cos^2(\theta_2) +\frac{1}{2} \sin(2\theta_2) \sin(2 \theta_4) cos(\alpha_4), 
$$
$$
(1\,3\,4) + (3\,1\,4) \leftrightarrow 2 \cos^2(\theta_3)\cos^2(\theta_4) +\frac{1}{2} \sin(2 \theta_3) \sin(2 \theta_4) cos(\alpha_3-\alpha_4),  
$$
$$
(2\,3\,4)+ (3\,2\,4)  \leftrightarrow 2 \cos^2(\theta_2) \cos^2(\theta_3) \cos^2(\theta_4) + 2 \sin^2(\theta_2) \sin^2(\theta_3) \sin^2(\theta_4) + 
\frac{1}{2} \sin(2 \theta_2) \sin(2 \theta_4) \cos(\alpha_4) + 
$$
$$
\frac{1}{2} \sin ( 2 \theta_2) \sin(2 \theta_3) \cos(\alpha_3) + \frac{1}{2} \sin(2 \theta_3) \sin( 2 \theta_4) \cos (\alpha_3-\alpha_4) 
$$

\subsection{Six 4-cycles}It is convenient to group the three $4$-cycles in pairs where one cycle is the inverse of the other so that the corresponding contributions are one the complex conjugate of the other, giving a  real value. The three four cycles pairs are $(1\,2\,3\,4)+(4\,3\,2\,1)$, $(1\,2\,4\,3)+(3\,4\,2\,1)$, $(1\,3\,2\,4)+(4\,2\,3\,1)$. Let us calculate the contribution of the pair  $(1\,2\,3\,4)+(4\,3\,2\,1)$, which is 
$$
2 Re\left( \langle \psi_1 |\psi_2 \rangle \langle \psi_2 | \psi_3 \rangle \langle \psi_3 | \psi_4 \rangle \langle \psi_4 | \psi_1 \rangle\right)= 
$$
$$
2 \cos^2(\theta_2) \cos^2(\theta_3) \cos^2(\theta_4) + \frac{\cos^2(\theta_2)}{2} \sin(2 \theta_4) \sin(2 \theta_3) \cos(\alpha_4-\alpha_3)+
$$
$$ \frac{\cos^2(\theta_4)}{2} \sin(2 \theta_2) \sin(2 \theta_3) \cos(\alpha_3)+ \frac{\sin^2(\theta_3)}{2} \sin(2 \theta_2) \sin(2\theta_4) \cos(\alpha_4). 
$$
We also have the contributions 
$$
(1\,2\,4\,3)+(3\,4\,2\,1) \leftrightarrow 2 \cos^2(\theta_2)\cos^2(\theta_3) \cos^2(\theta_4)+ \frac{\cos^2(\theta_2)}{2} \sin(2\theta_4) \sin(2\theta_3) \cos(\alpha_3-\alpha_4)+ 
$$
$$\frac{\cos^2(\theta_3)}{2} \sin(2 \theta_2) \sin(2\theta_4) \cos(\alpha_4)+ \frac{\sin^2(\theta_4)}{2} \sin(2\theta_2) \sin(2\theta_3) \cos(\alpha_3), 
$$
and
$$
(1\,3\,2\,4)+(4\,2\,3\,1) \leftrightarrow 2 \cos^2(\theta_2) \cos^2(\theta_3) \cos^2(\theta_4) + \frac{\cos^2(\theta_3)}{2} \sin(2 \theta_2) \sin(2\theta_4) \cos(\alpha_4)+
$$
$$
\frac{\cos^2(\theta_4)}{2} \sin(2\theta_2) \sin(2\theta_3)\cos(\alpha_3)+ \frac{\sin^2(\theta_2)}{2} \sin(2 \theta_3) \sin(2\theta_4) \cos(\alpha_4-\alpha_3)
$$
\subsection{Putting all together}
Let us sum all the contributions listed in the previous subsections. We shall write $4! \langle \psi | \Pi|\psi \rangle$ as 
$4! \langle \psi | \Pi|\psi \rangle =R+ A\left(\alpha_3, \alpha_4 \right)$ where $R$ and $A$ are the terms that do not, respectively, do depend on the $\alpha$'s, i.e.,  $\alpha_3$ and $\alpha_4$. Direct calculations and the use of standard trigonometric identities (by expressing everything in terms of the cosines only) lead to (use the shorthand notation $c_j:=\cos(\theta_j)$)
\begin{equation}\label{erre}
R=4-c_2^2-c_3^2-c_4^2+3c_2^2c_3^2+3c_2^2c_4^2+3c_3^2c_4^2+13c_2^2c_3^2c_4^2,
\end{equation}
\begin{equation}\label{A}
A=\cos(\alpha_3)(2+c_4^2)\sin(2\theta_2) \sin(2\theta_3)+\cos(\alpha_4)(2+c_3^2)\sin(2\theta_2)\sin(2\theta_4)+
\end{equation}
$$
 \cos(\alpha_3-\alpha_4)(2+c_2^2)\sin(2\theta_3) \sin(2\theta_4)
$$

\subsection{Minimizing with respect to $\alpha_4$'s}
We first minimize the part $A$ with respect to $\alpha_4$. Set 
\begin{equation}\label{abc}
a:=\sin(2\theta_2)\sin(2\theta_3)(2+c_4^2), \quad b:=\sin(2\theta_2) \sin(2\theta_4)(2+c_3^2), \quad c=\sin(2\theta_3) \sin(2 \theta_4)(2+c_2^2), 
\end{equation}
so that $A=a\cos(\alpha_3)+b \cos(\alpha_4)+c \cos(\alpha_3-\alpha_4)$. Minimizing with respect to $\alpha_4$,  using the fact that the minimum over $\alpha$ of 
$k_1 \cos(\alpha)+k_2 \sin(\alpha)$ is $-\sqrt{k_1^2+k_2^2}$  and standard trigonometric identities,  we can replace $A$ with 
$
a \cos(\alpha_3)-\sqrt{b^2+c^2+2bc \cos(\alpha_3)}
$. In our case, from (\ref{abc}),  $a$ and $bc$ have the same sign and since $\cos(\alpha_3)$ only multiplies these two terms we can assume without loss of generality that $a \geq 0$, $bc \geq 0$, and also replace $b$ with $|b|$ and $c$ with $|c|$. 
Notice also that the sign of $a$ does not influence the $R$ part of the function to minimize. 
In the following,  it is convenient to write $a$, $b$ and $c$ in (\ref{abc}) and the function to be minimized in terms of the variables $x:=c_2^2$, $y:=c_3^2$, $z:=c_4^2$, with $x,y,z \in [0,1]$, and the cases $\sin(2\theta_{2,3,4})=0$ corresponding  to $x,y,z=0,1$. 
With this definition, we write, $a=a(x,y,z)$, $b=b(x,y,z)$, and $c=c(x,y,z)$ as 
\begin{equation}\label{abcnew}
a:=4 \sqrt{(x-x^2)(y-y^2)}(2+z), \quad  b=4\sqrt{(x-x^2)(z-z^2)}(2+y), \quad c=4\sqrt{(y-y^2)(z-z^2)} (2+x). 
\end{equation}
We also write $\alpha_3:=2w \pi$, with $w \in [0,1]$.

\subsection{Numerical Minimization}
With the above notations and using the expression $R$ in (\ref{erre}), we need to minimize over $x,y,z,w$ in the box $[0,1]^4$, the function 
$$
f(x,y,z,w)=4-x-y-z+3xy+3xz+3yz+13xyz+a\cos(2w\pi)-\sqrt{b^2+c^2+2bc \cos(2w\pi)}. 
$$
We report below a MATLAB code that carries out such a minimization searching on a 
grid in $[0,1]^4$. Our program is an exhaustive search over a grid with stepsize $10^{-3}$. More sophisticated optimization methods can definitely be used although their use is  complicated by the fact that the function $f$ is not convex. We initialize the minimum of the function (the variable  `mini'  ) at the value achieved for $(x,y,z,w)=(0,0,0,0)$, which is $4$. 

\begin{verbatim}
mini=4
for x=0:0.001:1, 
  for y=0:0.001:1, 
   for z=0:0.001:1, 
    for w=0:0.001:1, 
      a=4*(2+z)*sqrt((x-x^2)*(y-y^2)); 
      b=4*(2+y)*sqrt((x-x^2)*(z-z^2)); 
      c=4*(2+x)*sqrt((y-y^2)*(z-z^2)); 
      curr=4-x-y-z+3*x*y+3*x*z+3*y*z+13*x*y*z+a*cos(2*w*pi)-sqrt(b^2+c^2+2*b*c*cos(2*w*pi)
      if curr < mini, 
         mini=curr
      end
   end
  end
 end
end
\end{verbatim}
The approximate minimum calculated with this routine is $\approx 1.3572$ therefore an approximate lower bound for $\||\psi\rangle \|_{SS}$ is $\approx \frac{1.3572}{24}$. 
\end{document}